\documentclass[a4paper,11pt]{article}
\usepackage{fullpage}
\usepackage{amsthm, amsfonts,mathrsfs,amssymb,amsmath,clrscode,multirow,graphicx}
\usepackage{graphicx}

\newcommand{\ket}[1]{\vert #1 \rangle}

\newcommand{\Nat}{\mathbb N}

\newcommand{\ceil}[1]{\lceil #1 \rceil}
\newcommand{\mybar}[1]{\lambda}

\newcommand{\poly}{\mathrm{poly}}

\newcommand{\Aa}{\mathscr{A}}

\newcommand{\Mm}{\mathscr{M}}

\newcommand{\sceil}[1]{\lceil #1 \rceil}

\newcommand{\Qc}{\mathcal{Q}}
\newtheorem{theorem}{Theorem}

\newtheorem{proposition}{Proposition}[section]
\newtheorem{definition}{Definition}[section]

\newtheorem{lemma}{Lemma}[section]
\newenvironment{proof-sketch}{\trivlist\item[]\emph{Brief proof sketch}:}%
{\unskip\nobreak\hskip 1em plus 1fil\nobreak$\Box$
\parfillskip=0pt%
\endtrivlist}
\begin{document}
\author{%
Fran{\c c}ois Le Gall \\ 
     Department of Computer Science\\
     Graduate School of Information Science and Technology\\
    The University of Tokyo\\
    \texttt{legall@is.s.u-tokyo.ac.jp} }
\title{A Time-Efficient Output-Sensitive Quantum Algorithm \\for  Boolean Matrix Multiplication}
\date{}
\maketitle
\begin{abstract}
This paper presents a quantum algorithm that computes the product of two $n\times n$ Boolean matrices in $\tilde O(n\sqrt{\ell}+\ell\sqrt{n})$ time, where~$\ell$ is the number of non-zero entries in the product. 
This improves the previous output-sensitive quantum algorithms for Boolean matrix multiplication in the time complexity setting by Buhrman and \v{S}palek (SODA'06) and Le Gall (SODA'12). 
We also show that our approach cannot be further improved unless a breakthrough is made: we prove that any significant improvement would imply the existence of an algorithm based on quantum search that multiplies two $n\times n$ Boolean matrices in $O(n^{5/2-\varepsilon})$ time, for some constant $\varepsilon>0$.
\end{abstract}
\section{Introduction}
\subsection{Background}
Multiplying two Boolean matrices, where addition is interpreted as a logical OR and multiplication as a logical AND, 
is a fundamental problem that has found applications in many areas of computer science
(for instance, computing the 
transitive closure of a graph \cite{Fischer+71,Furman70,MunroIPL71} or solving all-pairs path problems \cite{Dor+SICOMP00,Galil+97,SeidelJCSS95,Shoshan+FOCS99}).
The product of two $n\times n$ Boolean  matrices can be trivially computed in time $O(n^3)$.
The best known algorithm is obtained by seeing the input matrices as integer matrices, computing the product,
and converting the product matrix to a Boolean matrix. Using the algorithm by Coppersmith and Winograd~\cite{Coppersmith+90} for multiplying integer matrices (and more generally for multiplying matrices over any ring),
or its recent improvements by Stothers \cite{Stothers10} and Vassilevska Williams~\cite{WilliamsSTOC12},
this gives a classical algorithm for Boolean matrix multiplication with time complexity $O(n^{2.38})$. 

This algebraic 
approach has nevertheless many disadvantages, the main being that the huge constants involved in 
the complexities make these algorithms impractical. 
Indeed, in the classical setting, much attention has focused 
on algorithms that do not use reductions to matrix multiplication over rings, but instead are based on search or
on combinatorial arguments. Such algorithms are often called 
\emph{combinatorial algorithms}, and the main open problem in this field is to understand whether a $O(n^{3-\varepsilon})$-time 
combinatorial algorithm, for some constant $\varepsilon>0$,  exists for Boolean matrix multiplication.
Unfortunately, there has been little progress on this question. 
The best known combinatorial classical algorithm for Boolean 
matrix multiplication, by Bansal and Williams  \cite{Bansal+FOCS09}, 
has time complexity $O(n^3/\log^{2.25}(n))$. 

In the quantum setting, there exists a straightforward $\tilde O(n^{5/2})$-time\footnote{In this paper the notation $\tilde O$ suppresses $\poly(\log n)$ factors.} algorithm that
computes the
 product of two $n\times n$ Boolean matrices $A$ and~$B$: for each pair
of indexes 
$i,j\in\{1,2,\ldots,n\}$, check if there exists an index $k\in\{1,\ldots,n\}$ such that $A[i,k]=B[k,j]=1$ in time $\tilde O(\sqrt{n})$ using Grover's quantum search algorithm~\cite{GroverSTOC96}. Buhrman and \v{S}palek \cite{Buhrman+SODA06} observed that a similar approach leads to 
a quantum algorithm that computes the product $AB$ in $\tilde O(n^{3/2}\sqrt{\ell})$ time, 
where $\ell$ denotes the 
number of non-zero entries in $AB$. Since the parameter $\ell\in\{0,\ldots,n^2\}$ represents the 
sparsity of the output matrix, such an algorithm will be referred as \emph{output-sensitive}. 
Classical output-sensitive 
algorithms for Boolean matrix multiplication have also been constructed recently:
 Amossen and Pagh~\cite{Amossen+09} constructed an algorithm with time complexity $\tilde O(n^{1.724}\ell^{0.408}+n^{4/3}\ell^{2/3}+n^2)$, while 
Lingas \cite{Lingas11} constructed an algorithm with time complexity $\tilde O(n^2\ell^{0.188})$. 
The  above $\tilde O(n^{3/2}\sqrt{\ell})$-time quantum algorithm beats both of them when $\ell\le n^{1.602}$. 
Note that these two classical algorithms are based on the approach by Coppersmith and Winograd \cite{Coppersmith+90} and
are thus not combinatorial. 

Le Gall \cite{LeGallSODA12} has recently shown that 
there exists an output-sensitive quantum algorithm that computes the product of two $n\times n$ Boolean matrices
with time complexity $O(n^{3/2})$ if  $1\le \ell\le n^{2/3}$ and $O(n\ell^{3/4})$ if $n^{2/3}\le \ell\le n^{2}$.
This algorithm, which improves the quantum algorithm by Buhrman and \v{S}palek \cite{Buhrman+SODA06},
was constructed by combining ideas from works by Vassilevska Williams and Williams~\cite{Williams+FOCS10} and  
Lingas~\cite{Lingas11}.

Several developments concerning the quantum query complexity of this problem, 
where the complexity under consideration is the number of queries to the entries of the input matrices $A$ and $B$,
have also happened.
Output-sensitive quantum algorithms for Boolean matrix 
multiplication in the query complexity setting were first proposed in~\cite{Williams+FOCS10}, and then improved in~\cite{LeGallSODA12}.
Very recently, 
Jeffery, Kothari and Magniez~\cite{Jeffery+v2} significantly improved those results:
they showed that the quantum query complexity 
of computing the product of two $n\times n$ Boolean matrices with~$\ell$ non-zero entries is
$\tilde O(n\sqrt{\ell})$, and gave a matching (up to polylogarithmic factors) lower bound $\mathrm{\Omega}(n\sqrt{\ell})$.
The quantum query complexity of Boolean matrix multiplication may thus be considered as settled.

Can the quantum time complexity of Boolean matrix multiplication can be further improved as well?
The most fundamental question is of course whether there exists a quantum algorithm that
uses only quantum search or similar techniques with time complexity 
$O(n^{5/2-\varepsilon})$, for some constant $\varepsilon>0$, when $\ell\approx n^2$. This question is especially motivated by its apparently deep connection to the design of subcubic-time classical combinatorial algorithms for Boolean matrix multiplication: a $O(n^{5/2-\varepsilon})$-time quantum algorithm would correspond to an amortized cost of $O(n^{1/2-\varepsilon})$ per entry of the product, which may provides us with a new approach to develop a subcubic-time classical combinatorial algorithm, i.e., an algorithm with amortized cost of $O(n^{1-\varepsilon'})$ per entry of the product. Studying quantum algorithms for Boolean matrix multiplication in the time complexity setting can then, besides its own interest, be considered as a way to gain new insight about the optimal value of the exponent of matrix multiplication in the general case (i.e., for dense output matrices).
In comparison, when the output matrix is dense, the classical and the quantum query complexities of matrix multiplication are both trivially equal to $\mathrm{\Theta}(n^2)$.

\subsection{Statement of our results}
In this paper we build on the recent approach by Jeffery, Kothari and Magniez~\cite{Jeffery+v2}
to construct a new time-efficient  output-sensitive quantum algorithm for
Boolean matrix multiplication. 
Our main result is stated in the following theorem.
\begin{theorem}\label{theorem_1}
There exists a quantum algorithm that computes the product of two $n\times n$ Boolean matrices
with time complexity $\tilde O(n\sqrt{\ell}+\ell\sqrt{n})$,
where $\ell$ denotes the number of non-zero entries in the product. 
\end{theorem}

\begin{figure}\label{fig}
\begin{center}
 \includegraphics[scale=0.8]{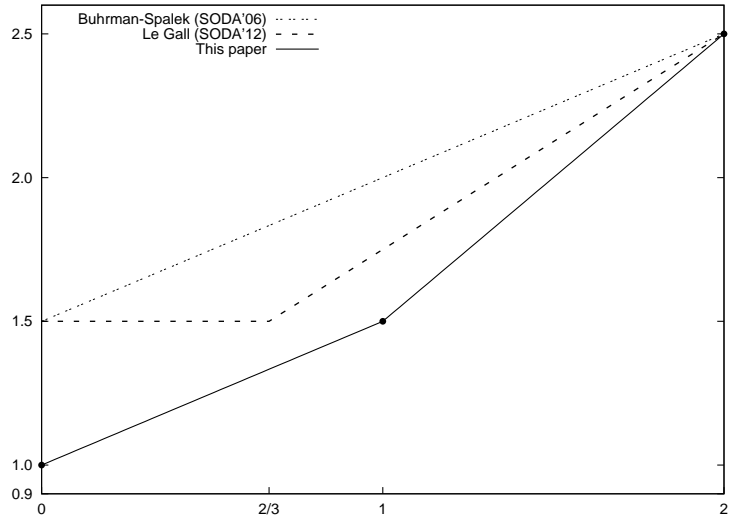}\vspace{-5mm}
\caption{The upper bounds on the time complexity of quantum algorithms for matrix multiplication given in Theorem \ref{theorem_1} (in solid line). 
The horizontal axis represents the logarithm of $\ell$ with respect to basis $n$ (i.e., the value $\log_n(\ell)$). 
The vertical axis represents the logarithm of the complexity with respect to basis $n$.
The dashed line represents the upper bounds on the time complexity obtained in \cite{LeGallSODA12},
and the dotted line represents the upper bounds obtained in \cite{Buhrman+SODA06}.
}
\end{center}
\end{figure}

The upper bounds of Theorem \ref{theorem_1} are illustrated in Figure 1. 
Our algorithm improves the quantum algorithm by Le Gall~\cite{LeGallSODA12}
for any value of~$\ell$ other than $\ell \approx n^2$ (we obtain the same 
upper bound $\tilde O(n^{2.5})$ for $\ell\approx n^2$).
It also beats 
the classical algorithms by Amossen and Pagh~\cite{Amossen+09} and 
Lingas \cite{Lingas11} mentioned earlier, which are based on the algebraic approach, for any value $\ell\le n^{1.847}$
(i.e., whenever $\ell\sqrt{n}\le n^2\ell^{0.188}$).

As will be explained in more details below, for $\ell\le n$ our result can be seen as a time-efficient version of the quantum algorithm constructed  for the query complexity setting in \cite{Jeffery+v2}. The query complexity lower bound $\mathrm{\Omega}(n\sqrt{\ell})$ proved in \cite{Jeffery+v2} shows that the time complexity of our algorithm is optimal, up to a possible polylogarithmic factor, for $\ell\le n$. 
The most interesting part of our results is perhaps the upper bound $\tilde O(\ell\sqrt{n})$ we obtain 
for $\ell\ge n$, which corresponds to the case where the output matrix  
is reasonably dense and differs from the query complexity upper bounds obtained in  \cite{Jeffery+v2}.
We additionally show that, for  values
$\ell\ge n$, no quantum algorithm based on search can perform better than ours 
unless there exists a quantum algorithm based on search 
that computes the product of two arbitrary $n\times n$ Boolean matrices
with time complexity significantly better that $n^{5/2}$.
The formal statement follows.
\begin{theorem}\label{theorem2}
Let $\delta$ be any function such that $\delta(n)>0$ for all $n\in\Nat^+$.
Suppose that, for some value $\lambda\ge n$, there exists 
a quantum algorithm $\Qc$ that,
given as input any $n\times n$ Boolean matrices $A$ and $B$ such that the number of non-zero entries 
in the product $AB$ is at most
$\lambda$, computes $AB$ in
$O\left(\lambda\sqrt{n}\cdot n^{-\delta(n)}\right)$ time.
Then there exists an algorithm using $\Qc$ as a black-box that
computes the product of two $n\times n$ Boolean matrices with overall time complexity
$\tilde O(n^{5/2-\delta(n)}+n)$.
\end{theorem} 
The reduction stated in Theorem \ref{theorem2} is actually classical and combinatorial:
the whole algorithm uses only classical combinatorial operations and calls to $\Qc$.
Thus Theorem~\ref{theorem2} implies that, if for a given value $\ell\ge n$ 
the complexity of Theorem~\ref{theorem_1} 
can be improved to  
$O\left(\ell\sqrt{n}/n^{\varepsilon}\right)$, for some constant $\varepsilon>0$, using
techniques similar to ours (i.e., based on quantum search),
then there exists an algorithm based on quantum search (and classical combinatorial operations)
that
computes the product of two $n\times n$ Boolean matrices with time complexity
$\tilde O(n^{5/2-\varepsilon}+n)$. 


\subsection{Overview of our techniques}
The main tool used to obtain our improvements is the new approach by Jeffery, Kothari and Magniez \cite{Jeffery+v2} to find collisions in the graph associated
with the multiplication of two $n\times n$ Boolean matrices. 
More precisely, it was shown in \cite{Jeffery+v2} how to find up to $t$ collisions in this graph, on a quantum computer, using
$\tilde O(\sqrt{nt}+\sqrt{\ell})$ queries, where $\ell$ is the number of non-zero entries in the product.
We construct (in Section \ref{sec_coll}) a time-efficient version of this algorithm that finds one collision in
$\tilde O(\sqrt{n}+\sqrt{\ell})$ time. We then use this algorithm to design 
a quantum algorithm that computes the matrix product in time $\tilde O(n\sqrt{\ell})$ when $\ell=O(n)$,
which proves Theorem \ref{theorem_1} for $\ell=O(n)$.
Our key technique is the introduction of a small data structure that is still powerful enough to enable 
time-efficient access to exactly all the information about the graph needed by the quantum searches. More precisely, while the size of the graph considered is $O(n^2)$, we show that the size of this data structure can be kept much smaller  --- roughly speaking, the idea is to keep a record of the non-edges of the graph. Moreover, the data structure is carefully chosen so that constructing it, at the beginning of the algorithm, can be done very efficiently (in $\tilde O(n)$ time), and updating it during the execution of the algorithm can be done at a cost less than the running time of the quantum searches. 



We then
prove 
 that the ability of finding up to $n$ non-zero entries of the matrix product  is enough by  
showing (in Section \ref{sec_red})
a classical reduction, for $\ell > n$, from the problem of computing the product of two $n\times n$ Boolean matrices with at most $\ell$ non-zero entries in the product
to the problem of computing $\ell/n$ separate products of two Boolean matrices, each product having at most $O(n)$ non-zero entries. 
The idea is to randomly permute the rows and columns of the input matrices in order
to make the output matrix homogeneous (in the sense that the non-zero entries are distributed almost uniformly), in which case we can decompose the input matrices into smaller blocks and ensure that each product of two smaller blocks contains, with non-negligible probability, at most $O(n)$ non-zero entries.
This approach is inspired by a technique introduced by Lingas~\cite{Lingas11} and then generalized in \cite{Jeffery+v1,LeGallSODA12}. The main difference is that here we focus on the number
of non-zero entries in the product of each pair of blocks, while  \cite{Jeffery+v1,LeGallSODA12,Lingas11} 
focused mainly on the size of the blocks.
The upper bounds of Theorem \ref{theorem_1} for $\ell\ge n$ follow directly from our reduction, and a stronger version of this reduction leads to the proof of Theorem~\ref{theorem2}.

\section{Preliminaries}\label{prelim}
In this paper we suppose that the reader is familiar with quantum computation, and especially with 
quantum search and its variants. We present below the model we are considering for accessing the 
input matrices on a quantum computer, and computing their product. This model is  
the same as the one used in \cite{Buhrman+SODA06,LeGallSODA12}.

Let $A$ and $B$ be two $n\times n$ Boolean matrices, for any positive integer $n$ (the model presented below can be 
generalized to deal with rectangular matrices in a straightforward way).
We suppose that these matrices can be accessed directly by a quantum algorithm.
More precisely, we have an oracle~$O_A$ that, for any $i,j\in\{1,\ldots,n\}$, 
any $a\in\{0,1\}$ and any  $z\in\{0,1\}^\ast$, performs the unitary mapping 
$
O_A\colon \ket{i}\ket{j}\ket{a}\ket{z}\mapsto\ket{i}\ket{j}\ket{a\oplus A[i,j]}\ket{z},
$
where $\oplus$ denotes the bit parity (i.e., the logical XOR).
We have a similar oracle $O_B$ for $B$. 
Since we are interested in time complexity, we will count all the computational steps of the algorithm and assign a cost of one for each call 
to $O_A$ or $O_B$, which corresponds to the cases where quantum access to the inputs $A$ and $B$ can be done at 
unit cost, for example in a random access model working in quantum superposition (we refer to  \cite{Nielsen+00} for an 
extensive treatment of such quantum random access memories).

Let $C=AB$ denote the product of the two matrices $A$ and $B$. 
Given any indices $i,j\in\{1,\ldots,n\}$ 
such that $C[i,j]=1$, a witness for this non-zero entry is defined as an 
index $k\in\{1,\ldots,n\}$ such that $A[i,k]=B[k,j]=1$. We define a quantum
algorithm for Boolean matrix multiplication as follows.

\begin{definition}
A quantum algorithm for Boolean matrix multiplication is a quantum algorithm that, 
when given access to oracles $O_A$ and $O_B$ corresponding to Boolean matrices 
$A$ and $B$, outputs
with probability at least $2/3$ all the non-zero entries of the product $AB$ along with
one witness for each non-zero entry. 
\end{definition}
The complexity of several algorithms in this paper will be stated using an upper bound $\lambda$ on the number $\ell$
of non-zero entries in the product $AB$.  The same complexity, up to a logarithmic factor, can actually be obtained 
even if no nontrivial upper bound on $\ell$ is known a priori. 
The idea is, similarly to what was done in \cite{Williams+FOCS10,LeGallSODA12}, to try successively $\lambda=2$ (and find up to 2 non-zero entries), $\lambda=4$ (and find up to 4 non-zero entries), \ldots 
and stop when no new non-zero entry is found. The complexity of this approach is, up to a logarithmic factor, the complexity of the last iteration
 (in which the value of $\lambda$ is $\lambda=2^{\ceil{\log_2\ell}+1}$ if $\ell$ is a power of two,
and $\lambda=2^{\ceil{\log_2\ell}}$ otherwise). 
In this paper we will then assume, without loss of generality, that a value $\lambda$ such that $\ell\le \lambda \le 2\ell$ is always
available. 

\section{Finding up to $O(n)$ Non-zero Entries}\label{sec_coll}
Let $A$ and $B$ be the two $n\times n$ Boolean matrices of which we want to compute the product.
In this section we define, following \cite{Jeffery+v2}, a graph collision problem and use it to show how to compute 
up to $O(n)$ non-zero entries of $AB$.

Let $G=(I,J,E)$ be a bipartite undirected graph over two disjoint sets~$I$ and~$J$, each of size~$n$.
The edge set $E$ is then a subset of $I\times J$.
When there is no ambiguity it will be convenient to write $I=\{1,\ldots, n\}$ and $J=\{1,\ldots, n\}$.
We now define the concept of a collision for the graph $G$.
\begin{definition}
For any index $k\in \{1,\ldots,n\}$, a $k$-collision for $G$ is an edge $(i,j)\in E$ such that 
$A[i,k]=B[k,j]=1$. A collision for $G$ is an edge $(i,j)\in E$ that is a $k$-collision for some index $k\in \{1,\ldots,n\}$.
\end{definition}

We suppose that the graph $G$ is given by 
a data structure $\Mm$ that contains the following information:
\begin{itemize}
\item[$\bullet$]
for each vertex $u$ in $I$, the degree of $u$; 
\item[$\bullet$]
for each vertex $u$ in $I$, a list of all the vertices of $J$ not connected to $u$.
\end{itemize}
The size of $\Mm$ is at most $\tilde O(n^2)$, but the key idea is that its size
will be much smaller when $G$ is ``close to'' a complete bipartite graph.
 Using adequate data structures to implement~$\Mm$ (e.g., using self-balancing binary search trees), 
we can perform the following four access operations in $\poly(\log n)$ time.
\begin{itemize}
\item[]
\texttt{get-degree($u$)}: get the degree of a vertex $u\in I$
\item[]
\texttt{check-connection($u,v$)}: check if the vertices $u\in I$ and $v\in J$ are connected
\item[]
\texttt{get-vert$_{I}$($r,d$)}:
get the $r$-th smallest vertex in $I$ that has degree at most $d$
\item[]
\texttt{get-vert$_{J}$($r,u$)}:
get the $r$-th smallest vertex in $J$ not connected to $u\in I$
\end{itemize}
For the latter two access operations, the order on the vertices refer to the usual 
order $\le$ obtained when seeing vertices in $I$ and $J$ as integers in $\{1,\ldots,n\}$.
We assume that these two access operations output an error message when the
query is not well-defined (i.e, when $r$ is too large).

Similarly, we can update $\Mm$ in $\poly(\log n)$ time to take in consideration the removal of one edge $(u,v)$ from $E$
(i.e., update the degree of $u$ 
and update the list of vertices not connected to~$u$). This low complexity 
will be crucial 
since our main algorithm (in Proposition \ref{proposition_2} below)
will 
remove successively edges from~$E$.

Let $L$ be an integer such that $0\le L\le n^2$. 
We will define our graph collision problem, denoted $\proc{Graph Collision($n$,$L$)}$,
as the problem of finding a collision for $G$ under the promise that $|E|\ge n^2-L$, 
i.e., there are at most $L$ missing edges in $G$. The formal definition is as follows.

\begin{codebox}
$\proc{Graph Collision($n$,$L$)}$  $\:\:[\:$ here $n\ge 1$ and $0\le L\le n^2$ $\:\:]$\\
\zi\const{input:} two $n\times n$ Boolean matrices $A$ and $B$
\zi \hspace{11mm} a bipartite graph $G=(I\cup J,E)$, with $|I|=|J|=n$, given by $\Mm$
\zi \hspace{11mm} an index $k\in \{1,\ldots,n\}$
\zi\const{promise:} $\:|E|\ge n^2-L$ 
\zi\const{output:} one $k$-collision if such a collision exists 
\end{codebox}

The following proposition shows that there exists a time-efficient quantum algorithm solving this problem.
The algorithm is similar to the query-efficient quantum algorithm given in~\cite{Jeffery+v2}, 
but uses the data structure $\Mm$ in order to keep the time complexity low.

\begin{proposition}\label{prop}
There exists a quantum algorithm running in time $\tilde O(\sqrt{L}+\sqrt{n})$ that solves, with high probability, the problem $\proc{Graph Collision($n$,$L$)}$.
\end{proposition}
\begin{proof}
We will say that a vertex $i\in I$ is marked if $A[i,k]=1$, 
and that a vertex $j\in J$ is marked if $B[k,j]=1$. Our goal is thus to find a pair $(i,j)\in E$ of marked vertices.
The algorithm is as follows.

We first use the minimum finding quantum algorithm from \cite{Durr+96} to find the marked vertex $u$ of largest degree in $I$, in $\tilde O(\sqrt{n})$ time using  \texttt{get-degree($\cdot$)} to obtain the order of a vertex from 
the data structure~$\Mm$.
 Let $d$ denote the degree of $u$, let $I'$ denote the set of vertices in $I$ with degree at most $d$,
and let $S$ denote the set of vertices in $J$ connected to $u$.
We then search for one marked vertex in $S$, using Grover's algorithm \cite{GroverSTOC96} with
\texttt{check-connection($u,\cdot$)}, in $\tilde O(\sqrt{n})$ time.
If we find one, then this gives us a $k$-collision and we end the algorithm. 
Otherwise we proceed as follows.
Note that, since each
 vertex in $I'$ has at most $d$ neighbors, by considering the number of missing edges we obtain:
$$|I'|\cdot (n-d)\le n^2-|E|\le  L.$$
Also note that $|J\backslash S|=n-d$.
We do a quantum search on $I'\times (J\backslash S)$ to find one pair of connected
marked vertices in time $\tilde O(\sqrt{|I'|\cdot|J\backslash S|})=\tilde O(\sqrt{L})$,
using \texttt{get-vert$_{I}$($\cdot,d$)} to access the vertices in $I'$ and 
\texttt{get-vert$_{J}$($\cdot,u$)} to access the vertices in  $J\backslash S$.
\end{proof}

We now show how an efficient quantum algorithm that computes up to $O(n)$ non-zero entries of the 
product of two $n\times n$ matrices can be constructed using Proposition \ref{prop}. 

\begin{proposition}\label{proposition_2}
Let $\lambda$ be a known value such that $\lambda=O(n)$.
Then there exists a quantum algorithm that,
given any $n\times n$ Boolean matrices $A$ and $B$ such that the number of non-zero entries 
in the product $AB$ is at most
$\lambda$, 
computes $AB$
in time  
$\tilde O(n\sqrt{\lambda})$.
\end{proposition}
\begin{proof}
Let $A$ and $B$ be two $n\times n$ Boolean matrices such that the product $AB$
has at most $\lambda$ non-zero entries.

We associate with this matrix multiplication the bipartite graph $G=(V,E)$, where $V=I\cup J$ with 
$I=J=\{1,\ldots, n\}$, and define the edge set as $E=I\times J$.
The two components $I$ and $J$ of $G$ are then fully connected: there is no missing edge.
It is easy to see that computing the product of $A$ and $B$ is equivalent 
to computing all the collisions, since a pair $(i,j)$ is a collision if and only if the entry 
in the $i$-th row and the $j$-th column of the product $AB$ is 1.

To find all the collisions, we will basically repeat the following approach: 
for a given $k$, search for a new $k$-collision in $G$ and remove the corresponding edge from $E$ by updating the data structure $\Mm$ corresponding to $G$. 
Since
we know that there are at most $\lambda$ non-zero entries in the matrix product $AB$, at most $\lambda$ collisions will be found (and then removed).
We are thus precisely interested in finding collisions when $|E|\ge n^2-\lambda$, i.e., when there are at most $\lambda$ missing edges in $G$.
We can then use the algorithm of Proposition \ref{prop}. 
The main subtlety is that we cannot simply try all the indexes $k$ successively since the cost would be too high. 
Instead, we will search for good indexes in a quantum way, as described in the next paragraph.

We partition the set of potential witnesses $K=\{1,\ldots,n\}$ into $m=\max(\lambda,n)$ subsets $K_1,\ldots,K_m$, 
each of size at most $\ceil{n/m}$. 
Starting with $s=1$, we repeatedly search for a pair~$(i,j)$ that is a $k$-collision for some $k\in K_s$. This is done by  
doing a Grover search over $K_s$ that invokes the algorithm of Proposition~\ref{prop}. 
Each time a new collision $(i,j)$ is found (which is a $k$-collision for some $k\in K_s$), we immediately remove 
the edge $(i,j)$ from $E$ by updating the data structure $\Mm$. When no other collision is found, we move to $K_{s+1}$.
We end the algorithm when the last set $K_m$ has been processed.

This algorithm will find, with high probability, all the collisions in the initial graph, and thus all the non-zero entries of $AB$.
Let us examine its time complexity.
We first discuss the complexity of creating the data structure $\Mm$ (remember that  updating $\Mm$ to
take in consideration  the removal of one edge from $E$ has polylogarithmic cost). 
Initially $|E|=n^2$, so each vertex of $I$ has the same degree~$n$.
Moreover, for each vertex $u\in I$, there is no vertex in $J$ not connected to~$u$. 
The cost for creating $\Mm$ is thus $\tilde O(n)$ time.
Next, we discuss the cost of the quantum search.
Let $\lambda_s$ denote the number of collisions found when examining the set $K_s$. Note that the search for collisions
(the Grover search that invokes the algorithm of Proposition \ref{prop})
is done $\lambda_s+1$ times 
when examining $K_s$ (we need one additional search to decide that there is no other collision). Moreover, we have $\lambda_1+\cdots+\lambda_m\le \lambda$.
The time complexity of the search is thus
$$
\tilde O\left(\sum_{s=1}^m\sqrt{|K_s|}\times (\sqrt{\lambda}+\sqrt{n}) \times (\lambda_s+1)\right)
=\tilde O\left(\sqrt{n}\lambda+n\sqrt{\lambda}\right)
=\tilde O\left(n\sqrt{\lambda}\right).
$$
The overall time complexity of the algorithm is thus $\tilde O(n\sqrt{\lambda}+n)=\tilde O(n\sqrt{\lambda})$.
\end{proof}

\section{Reduction to Several Matrix Multiplications}\label{sec_red}
Suppose that we have a randomized (or quantum) algorithm $\Aa$ that,
given any $m\times n$ Boolean matrix $A$ and any $n\times m$ Boolean matrix
$B$ such that the number of non-zero entries in the product $AB$ is known to be 
at most $L$, computes $AB$ with time complexity $T(m,n,L)$.
 For the sake of simplicity, we will make the following assumptions on $\Aa$:
\begin{itemize}
\item[(1)]
the time complexity of $\Aa$ does not exceed $T(m,n,L)$
even if the input matrices do not satisfy the promise (i.e., if there
are more than $L$ non-zero entries in the product);
\item[(2)]
the algorithm $\Aa$ never outputs that a zero entry of the product is non-zero; 
\item[(3)]
if the matrix product has at most $L$ non-zero entries, then with probability at least 
$1-1/n^3$ all these entries are found. 
\end{itemize}
These assumptions can be done without loss of generality when considering quantum algorithms for 
Boolean matrix multiplication as defined in Section~\ref{prelim}.
Assumption (1) can be guaranteed simply by supposing that the algorithm 
systematically stops after $T(m,n,L)$ steps. Assumption~(2) can 
be guaranteed since a witness is output for each potential 
non-zero entry found (the witness can be used to immediately check the result).
Assumption~(3) can be guaranted by repeating the original algorithm (which
has success probability at least 2/3) a logarithmic number of times.

The goal of this section is to show the following proposition.
\begin{proposition}\label{prop_dec}
Let $L$ be a known value such that $L\ge n$. Then, for any value $r\in\{1,\ldots,n\}$, there exists an algorithm that,
given any $n\times n$ Boolean matrices $A$ and $B$ such that the number of non-zero entries 
in the product $AB$ is at most~$L$, uses algorithm $\Aa$ to compute with high probability the product $AB$ in time
$$\tilde O\left(r^2\times T\left(\ceil{n/r},n,\frac{100(n+L/r)}{r}\right)+n\right).$$
\end{proposition}

We will need a lemma in order to prove Proposition \ref{prop_dec}. 
Let $C$ be an $n\times n$ Boolean matrix with at most $L$ non-zero entries. Let $r$ be a positive integer such that $r\le n$.
We choose an arbitrary partition $P_1$ of the set of rows of $C$ into~$r$ blocks in which each block has size at most $\ceil{n/r}$.
Similarly, we choose an arbitrary partition $P_2$ of the set of columns of~$C$ into $r$ blocks in which each block has size at most $\ceil{n/r}$.
These gives a decomposition of the matrix~$C$ into~$r^2$ subarrays (each of size at most $\ceil{n/r}\times \ceil{n/r}$).
We would like to say that, since $C$ has $L$ non-zero entries, then each subarray has at most $O(L/r^2)$ non-zero entries.
This is of course not true in general, but a similar statement will hold with high probability for a given subarray if we permute the rows and 
the columns of $C$ randomly.
We formalize this idea in the following lemma, which can be seen as an extension of a result proved 
by Lingas (Lemma 2 in \cite{Lingas11}).


\begin{lemma}\label{lemma_subarray}
Let $C$ be an $n\times n$ Boolean matrix with at most $L$ non-zero entries. 
Assume that $\sigma$ and $\tau$ are two permutations of the set $\{1,\ldots,n\}$ chosen independently uniformly at random.
Let $C[i,j]$ be any non-zero entry of $C$.
Then, with probability at least $9/10$, after permuting the rows according to $\sigma$ and the columns according to $\tau$,
the subarray containing this non-zero entry (i.e., the subarray containing the entry in the $\sigma(i)$-th row and the $\tau(j)$-th column)
has at most 
$$10\cdot\left(\frac{L\ceil{n/r}^2}{(n-1)^2}+2\ceil{n/r}-1\right)$$
 non-zero entries.
\end{lemma}
\begin{proof}
Since the permutations $\sigma$ and $\tau$ are chosen independently uniformly at random, we can 
consider that the values $\sigma(i)$ and $\tau(j)$ are first chosen, and that only after this choice the 
$2n-2$ other values (i.e., $\sigma(i')$ for $i'\neq i$ and $\tau(j')$ for $j'\neq j$) are chosen.

Assume that the values $\sigma(i)$ and $\tau(j)$ have been chosen.
Consider the subarray of $C$ containing the entry in the $\sigma(i)$-th row and the $\tau(j)$-th column.
Let $S$ denote the set of all the entries of the subarray, and let $T\subseteq S$
denote the set of all the entries of the subarray that are in the $\sigma(i)$-th row or in the $\tau(j)$-th column.
Note that $|S|\le \ceil{n/r}^2$ and $|T|\le 2\ceil{n/r}-1$. Let $X_s$, for each $s\in S$, denote the
random variable with value one if the entry $s$ of the subarray is one,
and value zero otherwise. The random variable representing the number of non-zero entries in the subarray is 
thus $Y=\sum_{s\in S} X_s$. The expectation of $Y$ is 
\begin{eqnarray*}
E[Y]=\sum_{s\in S} E[X_s]=\sum_{s\in S\backslash T} E[X_s]+\sum_{s\in T} E[X_s]
&\le& |S\backslash T|\times \frac{L}{(n-1)^2}+|T|\\
&\le&\frac{L\ceil{n/r}^2}{(n-1)^2}+2\ceil{n/r}-1.
\end{eqnarray*}
The first inequality is obtained by using the inequality $E[X_s]\le 1$ for each entry $s\in T$, and by noting that each of the non-zero entries of $C$ that are neither
in the $i$-th row nor in the $j$-th column has probability exactly $\frac{1}{(n-1)^2}$ to be moved  
into a given entry in $S\backslash T$ by the permutations of the remaining $n-1$ rows and $n-1$ columns. 
From Markov's inequality we obtain:
$$\Pr\left[Y\ge 10\cdot\left(\frac{L\ceil{n/r}^2}{(n-1)^2}+2\ceil{n/r}-1\right)\right]\le \frac{1}{10}.$$

The statement of the lemma follows from the observation that the 
above inequality holds for any choice of $\sigma(i)$ and $\tau(j)$.
\end{proof}

\begin{proof}[Proof of Proposition \ref{prop_dec}]
Take two arbitrary partitions $P_1,P_2$ of the set $\{1,\ldots,n\}$ into $r$ blocks in which each block has size at most $\ceil{n/r}$.
Let us write 
$$\Delta=10\cdot\left(\frac{L\ceil{n/r}^2}{(n-1)^2}+2\ceil{n/r}-1\right).$$
It is easy to show that $\Delta\le \frac{100\left(n+L/r\right)}{r}$ when $n\ge 3$.
We will repeat the following procedure $\ceil{c\log n}$ times, for some large enough constant $c$: 
\begin{itemize}
\item[1.]
Permute the rows of $A$ randomly and denote by $A^\ast$ the resulting matrix;\\
Permute the columns of $B$ randomly and denote by $B^\ast$ the resulting matrix;
\item[2.]
Decompose $A^\ast$ into $r$ smaller matrices $A^\ast_1,\ldots,A^\ast_r$ of size at most $\ceil{n/r}\times n$ by 
partitioning the rows of $A^\ast$ according to $P_1$;\\
Decompose $B^\ast$ into $r$ smaller matrices $B^\ast_1,\ldots,B^\ast_r$ of size at most $n\times \ceil{n/r}$ by 
partitioning the columns of $B^\ast$ according to $P_2$;
\item[3.]
For each $s\in\{1,\ldots,r\}$ and each $t\in\{1,\ldots,r\}$, compute up to $\Delta$ non-zero entries of the product $A^\ast_sB^\ast_t$ using the algorithm $\Aa$.
\end{itemize}

The time complexity of this procedure is 
$\tilde O\left(r^2\times T\left(\ceil{n/r},n,\Delta\right)+n\right)$, 
where the additive term $\tilde O(n)$ represents the time complexity of dealing with the permutations of rows and columns (note that $A^\ast$, $B^\ast$, the $A^\ast_s$'s and the $B^\ast_t$'s have not to be computed explicitly; we only need to be able to recover in polylogarithmic time a given entry of these matrices).

We now show the correctness of the algorithm. First, the algorithm will never output that a zero entry of $AB$ is non-zero, from our assumption on the algorithm~$\Aa$. Thus all the entries output by the algorithm are non-zero entries of $AB$. The question is
whether all the non-zero entries are output. 

Let $(i,j)$ be a fixed non-zero entry of $AB$. Note that each matrix product $A^\ast_sB^\ast_t$ corresponds to 
a subarray of $A^\ast B^\ast$.
From our assumptions on algorithm
$\Aa$, 
this entry will be output with probability $p\ge 1-1/n^3$ at Step 3 of the procedure if the entry (after permutation of the rows and the columns) is in a subarray of $A^\ast B^\ast$ containing at most $\Delta$ non-zero entries. From Lemma \ref{lemma_subarray}, this happens with probability at least $9/10$. 
With probability at least $1-(1/10)^{\ceil{c\log n}}$ this case will happen at least once during the $\ceil{c\log n}$ iterations of the procedure.
By choosing the constant $c$ large enough, we have $1-(1/10)^{\ceil{c\log n}}\ge 1-1/n^3$, and then 
the algorithm outputs this non-zero entry with probability at least $(1-1/n^3)p\ge 1-2/n^3$.
By the union bound we conclude 
that the probability that the algorithm outputs all the non-zero entries of $AB$ is at least $1-2/n$.
\end{proof}

\section{Proofs of Theorems \ref{theorem_1} and \ref{theorem2}}

In this section we give the proofs of Theorems \ref{theorem_1} and \ref{theorem2}.
\begin{proof}[Proof of Theorem \ref{theorem_1}]
Let $A$ and $B$ be two $n\times n$ Boolean matrices such that the product $AB$ has $\ell$ non-zero entries. 
Remember that, as discussed in Section~\ref{prelim},  
an integer $\lambda\in\{1,\ldots,n^2\}$ such that $\ell\le \lambda\le 2\ell$ is known.

If $\lambda\le n$ then the product $AB$ can be computed in time $\tilde O(n\sqrt{\lambda})=\tilde O(n\sqrt{\ell})$ 
by the algorithm of  Proposition \ref{proposition_2}.
Now consider the case $n\le \lambda\le n^2$.
By Proposition~\ref{prop_dec} (with the value $r=\sceil{\sqrt{\lambda/n}}$), the product of $A$ and $B$ can be 
computed  with complexity
$
\tilde O\left(\frac{\lambda}{n}\times T\left(n,n,\Delta\right)+n\right),
$
where $\Delta=O(n)$.
Combined with Proposition~\ref{proposition_2}, this gives
a quantum algorithm that computes the product $AB$ in  
$\tilde O\left(\lambda\sqrt{n}\right)=\tilde O\left(\ell\sqrt{n}\right)$ time. 
\end{proof}
\begin{proof}[Proof of Theorem \ref{theorem2}]
Suppose the existence of a quantum algorithm that computes 
in time 
\[
O\left(\lambda\sqrt{n}\cdot n^{-\delta(n)}\right)
\]
the product of any two $n\times n$ Boolean matrices such that the number of 
non-zero entries in their product is at most~$\lambda$. 
Let $c$ be a positive constant. 
Using Proposition~\ref{prop_dec} with the values $r=\sceil{cn/\sqrt{\lambda}}$ and $L=n^2$, we obtain a 
quantum algorithm that computes 
the product of two $n\times n$ Boolean matrices in time
$$
\tilde O\left(\frac{n^2}{\lambda}\times T\left(n,n,\frac{100n}{\sceil{cn/\sqrt{\lambda}}}+\frac{100n^2}{\sceil{cn/\sqrt{\lambda}}^2}\right)+n\right).
$$
 By choosing the constant $c$ large enough, we can rewrite this upper bound as
\[
\tilde O\left(\frac{n^2}{\lambda}\times T\left(n,n,\lambda\right)+n\right)=
\tilde O\left(n^{5/2-\delta(n)}+n\right),
\]
which concludes the proof of the theorem.
\end{proof}

\section*{Acknowledgments}
The author is grateful to Stacey Jeffery, Robin Kothari and Fr{\'e}d{\'e}ric Magniez
for helpful discussions and comments, and for communicating to him preliminary versions of their
works.
He also
acknowledges support 
from the JSPS and the MEXT,
under the grant-in-aids 
Nos.~24700005, 24106009 and 24240001.

\end{document}